\documentclass{article}%
\usepackage{amsmath}
\usepackage{amsfonts}
\usepackage{amssymb}
\usepackage{graphicx}%
\providecommand{\U}[1]{\protect\rule{.1in}{.1in}}
\newtheorem{theorem}{Theorem} [section]

\newtheorem{conjecture}[theorem]{Conjecture}
\newtheorem{corollary}[theorem]{Corollary}

\newtheorem{definition}[theorem]{Definition}

\newtheorem{lemma}[theorem]{Lemma}

\newtheorem{problem}[theorem]{Problem}
\newtheorem{proposition}[theorem]{Proposition}
\newtheorem{remark}[theorem]{Remark}

\newenvironment{proof}[1][Proof]{\noindent\textbf{#1.} }{\ \rule{0.5em}{0.5em}}
\begin{document}

\author{Vadim E. Levit\\Department of Computer Science\\Ariel University, Israel\\levitv@ariel.ac.il
\and Eugen Mandrescu\\Department of Computer Science\\Holon Institute of Technology, Holon, Israel\\eugen\_m@hit.ac.il}
\title{Critical sets, crowns and local maximum independent sets}
\date{}
\maketitle

\begin{abstract}
A set $S\subseteq V(G)$ is \textit{independent }(or\textit{ stable}) if no two
vertices from $S$ are adjacent, and by $\mathrm{Ind}(G)$ we mean the set of
all independent sets of $G$.

A set $A\in\mathrm{Ind}(G)$ is \textit{critical} (and we write $A\in
CritIndep(G)$) if $\left\vert A\right\vert -\left\vert N(A)\right\vert
=\max\{\left\vert I\right\vert -\left\vert N(I)\right\vert :I\in
\mathrm{Ind}(G)\}$ \cite{Zhang1990}, where $N(I)$ denotes the neighborhood of
$I$.

If $S\in\mathrm{Ind}(G)$ and there is a matching from $N(S)$ into $S$, then
$S$ is a \textit{crown }\cite{AFLS2007}, and we write $S\in Crown(G)$.

Let $\Psi(G)$ be the family of all local maximum independent sets of graph
$G$, i.e., $S\in\Psi(G)$ if $S$ is a maximum independent set in the subgraph
induced by $S\cup N(S)$ \cite{LevMan2}.

In this paper we show that $CritIndep(G)\subseteq Crown(G)$ $\subseteq\Psi(G)$
are true for every graph. In addition, we present some classes of graphs where
these families coincide and form greedoids or even more general set systems
that we call augmentoids.


\end{abstract}

\section{Introduction}

Throughout this paper $G$ is a finite simple graph with vertex set $V(G)$ and
edge set $E(G)$. If $X\subseteq V\left(  G\right)  $, then $G[X]$ is the
subgraph of $G$ induced by $X$. By $G-W$ we mean either the subgraph
$G[V\left(  G\right)  -W]$, if $W\subseteq V(G)$, or the subgraph obtained by
deleting the edge set $W$, for $W\subseteq E(G)$. The \textit{neighborhood}
$N(v)$ of a vertex $v\in V\left(  G\right)  $ is the set $\{w:w\in V\left(
G\right)  $ \textit{and} $vw\in E\left(  G\right)  \}$, while the
\textit{closed neighborhood} $N[v]$\ of $v\in V\left(  G\right)  $ is the set
$N(v)\cup\{v\}$; in order to avoid ambiguity, we use also $N_{G}(v)$ instead
of $N(v)$. The \textit{neighborhood} $N(A)$ of $A\subseteq V\left(  G\right)
$ is $\{v\in V\left(  G\right)  :N(v)\cap A\neq\emptyset\}$, and
$N[A]=N(A)\cup A$. We may also use $N_{G}(A)$ and $N_{G}\left[  A\right]  $,
when referring to neighborhoods in a graph $G$. A vertex $v$ is
\textit{isolated} if $N(v)=\emptyset$. Let us define $\mathrm{isol}(G)$ as the
set of all isolated vertices. If $A,B\subset V(G),A\cap B=\emptyset$, then by
$\left(  A,B\right)  $ is denoted the set $\left\{  ab:ab\in E(G),a\in A,b\in
B\right\}  $.

The graph $G$ is \textit{unicyclic} if it is connected and has a unique cycle.

Let $V(G)=\{v_{i}:1\leq i\leq n\}$. Joining each $v_{i}$ to all the vertices
of a copy of a graph $H$, we obtain a new graph, called the \textit{corona} of
$G$ and $H$, denoted by $G\circ H$.

A set $S\subseteq V(G)$ is \textit{independent} (or \textit{stable}) if no two
vertices from $S$ are adjacent, and by $\mathrm{Ind}(G)$ we mean the family of
all the independent sets of $G$. An independent set of maximum size is a
\textit{maximum independent set} of $G$, and the \textit{independence number
}$\alpha(G)$ of $G$ is $\max\{\left\vert S\right\vert :S\in\mathrm{Ind}(G)\}$.
Let $\Omega(G)$ denote the family of all maximum independent sets, and
$\mathrm{core}(G)=%
{\displaystyle\bigcap}
\{S:S\in\Omega(G)\}$ \cite{HamHanSim,LevMan2002a}. The problem of whether
$\mathrm{core}(G)\neq\emptyset$ is \textbf{NP}-hard \cite{BorosGolLev}.

A \textit{matching} in a graph $G$ is a set of edges $M\subseteq E\left(
G\right)  $ such that no two edges of $M$ share a common vertex. A
\textit{maximum matching} is a matching of maximum cardinality. By $\mu(G)$ is
denoted the cardinality of a maximum matching. A matching is \textit{perfect}
if it saturates all the vertices of the graph. A matching $M=\{a_{i}%
b_{i}:a_{i},b_{i}\in V(G),1\leq i\leq k\}$ is \textit{a uniquely restricted
matching} if $M$ is the unique perfect matching of $G[\{a_{i},b_{i}:1\leq
i\leq k\}]$ \cite{GolHiLew}.

Recall that if $\alpha(G)+\mu(G)=\left\vert V(G)\right\vert $, then $G$ is a
\textit{K\"{o}nig-Egerv\'{a}ry graph} \cite{Deming1979,Sterboul1979}. As a
well-known example, each bipartite graph is a K\"{o}nig-Egerv\'{a}ry graph as
well. Various properties of K\"{o}nig-Egerv\'{a}ry graphs can be found in
\cite{Bonomo2013,JarLevMan2016,JarLevMan2015b,Korach2006,LevMan2003,LevMan2006,LevMan2007,LevMan2011,LevManLemma2011,LovPlum1986}%
.

If $S$ is an independent set of a graph $G$ and $A=V\left(  G\right)  -S$,
then we write $G=S\ast A$. Evidently, each graph admits such representations.
For instance, if $E(A)=\emptyset$, then $G=S\ast A$ is bipartite; if $H$ is
complete, then $G=S\ast A$ is a split graph.

\begin{theorem}
\label{ThKE}\cite{LevMan2002a} Let $G$ be a connected graph. Then the
following are equivalent:

\emph{(i)} $G$ is a \textit{K\"{o}nig-Egerv\'{a}ry} graph;

\emph{(ii)} $G=S\ast A$, where $S\in\Omega(G)$ and $\left\vert S\right\vert
\geq\mu(G)=\left\vert A\right\vert $;

\emph{(iii)} $G=H_{1}\ast H_{2}$, where $V(H_{1})=S$ is independent,
$\left\vert S\right\vert \geq\left\vert V\left(  H_{2}\right)  \right\vert $,
and $\left(  S,V\left(  H_{2}\right)  \right)  $ contains a matching $M$ with
$\left\vert M\right\vert =\left\vert V\left(  H_{2}\right)  \right\vert $.
\end{theorem}

Let us notice that a disconnected graph is a K\"{o}nig-Egerv\'{a}ry graph if
and only if each of its connected components induces a K\"{o}nig-Egerv\'{a}ry graph.

\begin{theorem}
\label{ThKEMatch}\cite{LevMan2013b} For a graph $G$ the following properties
are equivalent:

\emph{(i)} $G$ is a K\"{o}nig-Egerv\'{a}ry graph;

\emph{(ii)} there is $S\in%
\Omega
(G)$ such that each maximum matching of $G$ is contained in $(S,V(G)-S)$;

\emph{(iii)} for every $S\in%
\Omega
(G)$ each maximum matching of $G$ is contained in $(S,V(G)-S)$.
\end{theorem}

For $X\subseteq V(G)$, the number $\left\vert X\right\vert -\left\vert
N(X)\right\vert $ is the \textit{difference} of $X$, denoted $d(X)$. The
\textit{critical difference} $d(G)$ is $\max\{d(X):X\subseteq V(G)\}$. The
number $\max\{d(I):I\in\mathrm{Ind}(G)\}$ is the \textit{critical independence
difference} of $G$, denoted $id(G)$. Clearly, $d(G)\geq id(G)$. It was shown
in \cite{Zhang1990} that $d(G)$ $=id(G)$ holds for every graph $G$. If $A$ is
an independent set in $G$ with $d\left(  X\right)  =d(G)$, then $A$ is a
\textit{critical independent set}. A \textit{maximum critical independent }set
is a critical independent set of maximum cardinality. Let
\begin{align*}
CritIndep(G)  &  =\left\{  S:S\text{ \textit{is a critical independent set in}
}G\right\}  \text{ and}\\
MaxCritIndep(G)  &  =\left\{  S:S\text{ \textit{is a maximum critical
independent set in} }G\right\}  .
\end{align*}

For example, consider the graph $G$ from Figure \ref{fig511}. Note that
$X=\{v_{1},v_{2},v_{3},v_{4}\}$ is a critical set, since $N(X)=\{v_{3}%
,v_{4},v_{5}\}$ and $d(X)=1=d(G)$, while $S=\{v_{1},v_{2},v_{3},v_{6},v_{7}\}$
is a critical independent set, because $d(S)=1=d(G)$. Notice that
$\{v_{1},v_{2},v_{3},v_{6},v_{7},v_{10}\}\in$ \textrm{MaxCritIndep}$(G)$.
Other critical sets are $\{v_{1},v_{2}\}$, $\{v_{1},v_{2},v_{3}\}$,
$\{v_{1},v_{2},v_{3},v_{4},v_{6},v_{7}\}$. 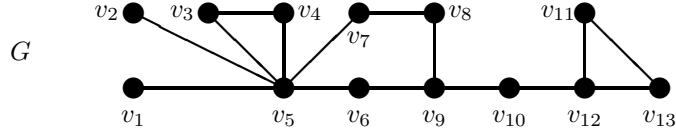
\begin{figure}[h]
\setlength{\unitlength}{1cm}\begin{picture}(5,1.9)\thicklines
\multiput(6,0.5)(1,0){6}{\circle*{0.29}}
\multiput(5,1.5)(1,0){4}{\circle*{0.29}}
\multiput(4,0.5)(0,1){2}{\circle*{0.29}}
\put(10,1.5){\circle*{0.29}}
\put(4,0.5){\line(1,0){7}}
\put(4,1.5){\line(2,-1){2}}
\put(5,1.5){\line(1,-1){1}}
\put(5,1.5){\line(1,0){1}}
\put(6,0.5){\line(0,1){1}}
\put(6,0.5){\line(1,1){1}}
\put(7,1.5){\line(1,0){1}}
\put(8,0.5){\line(0,1){1}}
\put(10,0.5){\line(0,1){1}}
\put(10,1.5){\line(1,-1){1}}
\put(4,0.1){\makebox(0,0){$v_{1}$}}
\put(3.65,1.5){\makebox(0,0){$v_{2}$}}
\put(4.65,1.5){\makebox(0,0){$v_{3}$}}
\put(6.35,1.5){\makebox(0,0){$v_{4}$}}
\put(6,0.1){\makebox(0,0){$v_{5}$}}
\put(7,0.1){\makebox(0,0){$v_{6}$}}
\put(7,1.15){\makebox(0,0){$v_{7}$}}
\put(8,0.1){\makebox(0,0){$v_{9}$}}
\put(8.35,1.5){\makebox(0,0){$v_{8}$}}
\put(9.65,1.5){\makebox(0,0){$v_{11}$}}
\put(9,0.1){\makebox(0,0){$v_{10}$}}
\put(10,0.1){\makebox(0,0){$v_{12}$}}
\put(11,0.1){\makebox(0,0){$v_{13}$}}
\put(2.5,1){\makebox(0,0){$G$}}
\end{picture}\caption{\textrm{core}$(G)=\{v_{1},v_{2},v_{6},v_{10}\}$ is a
critical set.}%
\label{fig511}%
\end{figure}

Critical independent sets are of interest for both practical and theoretical
purposes. Zhang proved that a critical independent set can be found in
polynomial time \cite{Zhang1990}. A simpler algorithm, reducing the critical
independent set problem to computing a maximum independent set in a bipartite
graph is given in \cite{Ageev}.

\begin{lemma}
\label{MatchLemma}\cite{Larson2007} There is a matching from $N(S)$ into $S$
for every critical independent set $S$.
\end{lemma}

A proof of a conjecture of Graffiti.pc \cite{DeLaVina} yields a new
characterization of K\"{o}nig-Egerv\'{a}ry graphs: these are exactly the
graphs, where there exists a critical maximum independent set
\cite{Larson2011}. 

It is known that $d(G)\geq$ $\alpha(G)-\mu(G)$ holds for every graph
\cite{LevMan2012a,Lorentzen1966,Schrijver2003}. In \cite{LevMan2012b} it is
proved the following.

\begin{theorem}
\label{th8}\cite{LevMan2012b} For a K\"{o}nig-Egerv\'{a}ry graph $G$ the
following equalities hold
\[
d(G)=\left\vert \mathrm{core}(G)\right\vert -\left\vert N(\mathrm{core}%
(G))\right\vert =\alpha(G)-\mu(G).
\]

\end{theorem}

Using this finding, we have strengthened the characterization from
\cite{Larson2011}.

\begin{theorem}
\label{th5}\cite{LevMan2012b} $G$ is a K\"{o}nig-Egerv\'{a}ry graph if and
only if each of its maximum independent sets is critical.
\end{theorem}

For a graph $G$, let us denote
\[
\mathrm{\ker}(G)=%
{\displaystyle\bigcap}
\left\{  A:A\text{ \textit{is a critical independent set}}\right\}  \text{
\cite{LevMan2012a}.}%
\]

Clearly, $\mathrm{isol}(G)\subseteq\mathrm{\ker}(G)$. Some more on the
internal structure of $\mathrm{\ker}(G)$ may be found in \cite{LevMan2013a}.

\begin{theorem}
\cite{LevMan2012a} \emph{(i) }For every graph $G$, $\mathrm{\ker}%
(G)\subseteq\mathrm{core}(G)$.

\cite{LevMan2011b} \emph{(ii) }If $G$ is a bipartite graph, then
$\mathrm{\ker}(G)=\mathrm{core}(G)$.

\cite{LevMan2014a} \emph{(iii) }If $G$ is a unicyclic
non-K\"{o}nig-Egerv\'{a}ry graph, then $\mathrm{\ker}(G)=\mathrm{core}(G)$.
\end{theorem}

Since $\left\vert \mathrm{\ker}(G)\right\vert \neq1$ holds for any graph $G$
without isolated vertices \cite{LevMan2019}, it follows that if $\left\vert
\mathrm{isol}(G)\right\vert \neq1$, then $\left\vert \mathrm{core}%
(G)\right\vert \neq1$ for bipartite graphs \cite{Jamison1987,LevMan2002a} and
unicyclic non-K\"{o}nig-Egerv\'{a}ry graphs \cite{LevMan2012e,LevMan2014a}.

Notice that there are non-bipartite graphs enjoying the equality
$\mathrm{\ker}(G)=\mathrm{core}(G)$; e.g., the graphs from Figure \ref{fig14},
where only $G_{1}$ is a K\"{o}nig-Egerv\'{a}ry graph. 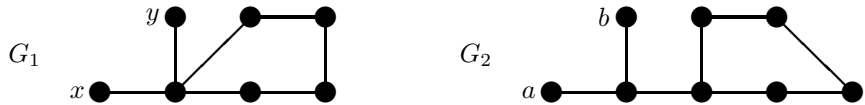
\begin{figure}[h]
\setlength{\unitlength}{1cm}\begin{picture}(5,1.2)\thicklines
\multiput(2,0)(1,0){4}{\circle*{0.29}}
\multiput(3,1)(1,0){3}{\circle*{0.29}}
\put(2,0){\line(1,0){3}}
\put(3,0){\line(0,1){1}}
\put(5,0){\line(0,1){1}}
\put(4,1){\line(1,0){1}}
\put(3,0){\line(1,1){1}}
\put(1.7,0){\makebox(0,0){$x$}}
\put(2.7,1){\makebox(0,0){$y$}}
\put(1,0.5){\makebox(0,0){$G_{1}$}}
\multiput(8,0)(1,0){5}{\circle*{0.29}}
\multiput(9,1)(1,0){3}{\circle*{0.29}}
\put(8,0){\line(1,0){4}}
\put(9,0){\line(0,1){1}}
\put(10,0){\line(0,1){1}}
\put(10,1){\line(1,0){1}}
\put(11,1){\line(1,-1){1}}
\put(7.7,0){\makebox(0,0){$a$}}
\put(8.7,1){\makebox(0,0){$b$}}
\put(7,0.5){\makebox(0,0){$G_{2}$}}
\end{picture}\caption{$\mathrm{core}(G_{1})=\ker\left(  G_{1}\right)
=\{x,y\}$ and $\mathrm{core}(G_{2})=\ker\left(  G_{2}\right)  =\{a,b\}$.}%
\label{fig14}%
\end{figure}

There is a non-bipartite K\"{o}nig-Egerv\'{a}ry graph $G$, such that
$\mathrm{\ker}(G)\neq\mathrm{core}(G)$. For instance, the graph $G_{1}$ from
Figure \ref{fig222} has $\mathrm{\ker}(G_{1})=\left\{  x,y\right\}  $, while
$\mathrm{core}(G_{1})=\left\{  x,y,u,v\right\}  $. The graph $G_{2}$ from
Figure \ref{fig222} has $\mathrm{\ker}(G_{2})=\emptyset$, while $\mathrm{core}%
(G_{2})=\left\{  w\right\}  $.

\begin{figure}[h]
\setlength{\unitlength}{1cm}\begin{picture}(5,1.3)\thicklines
\multiput(4,0)(1,0){4}{\circle*{0.29}}
\multiput(3,1)(1,0){5}{\circle*{0.29}}
\put(4,0){\line(1,0){3}}
\put(4,0){\line(0,1){1}}
\put(3,1){\line(1,-1){1}}
\put(5,0){\line(0,1){1}}
\put(5,0){\line(1,1){1}}
\put(5,1){\line(1,-1){1}}
\put(6,0){\line(0,1){1}}
\put(7,0){\line(0,1){1}}
\put(2.7,1){\makebox(0,0){$x$}}
\put(3.7,1){\makebox(0,0){$y$}}
\put(4.7,1){\makebox(0,0){$u$}}
\put(6.3,1){\makebox(0,0){$v$}}
\put(2,0.5){\makebox(0,0){$G_{1}$}}
\multiput(10,0)(1,0){2}{\circle*{0.29}}
\multiput(12,0)(0,1){2}{\circle*{0.29}}
\put(10,0){\line(1,0){2}}
\put(11,0){\line(1,1){1}}
\put(12,0){\line(0,1){1}}
\put(10,0.3){\makebox(0,0){$w$}}
\put(9,0.5){\makebox(0,0){$G_{2}$}}
\end{picture}\caption{Both $G_{1}$ and $G_{2}$\ are K\"{o}nig-Egerv\'{a}ry
graphs. Only $G_{2}$\ has a perfect matching.}%
\label{fig222}%
\end{figure}
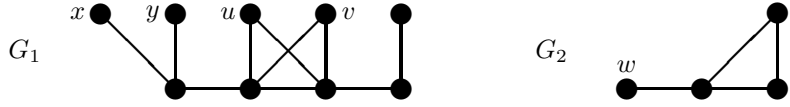

Recall that if $S\in\mathrm{Ind}(G)$ and there is a matching from $N(S)$ into
$S$, then $S$ is a \textit{crown} in $G$, of order $\left\vert S\right\vert
+\left\vert N(S)\right\vert $ \cite{AFLS2007,ChCh2008,DowneyFellows1999}. If
$\left\vert S\right\vert =\left\vert N(S)\right\vert $, then $S$ is a
\textit{straight crown}. A crown of maximum cardinality is a \textit{maximum
crown}. Let
\begin{align*}
Crown(G) &  =\left\{  S:S\text{ \textit{is a crown of} }G\right\}  \text{ and
}\\
MaxCrown(G) &  =\left\{  S:S\text{ \textit{is a maximum crown of} }G\right\}
.
\end{align*}
It is convenient to consider that $\emptyset\in$ $Crown(G)$, and also
$\left\{  v\right\}  \in$ $Crown(G)$ for every isolated vertex $v$.

\begin{theorem}
\label{th2}\cite{AFLS2007} Isolating a crown of maximum order is solvable in
polynomial time.
\end{theorem}

Let us emphasize the significance of the above theorem in the context of the
following result.

\begin{theorem}
\cite{Sloper2005} Given a graph $G$ and an integer $k$, it is \textbf{NP}%
-complete to decide whether $G$ contains a crown whose order is exactly $k$.
\end{theorem}

A set $A\subseteq V(G)$ is a \textit{local maximum independent (stable) set}
of $G$ if $A$ is a maximum independent set in the subgraph $G[N[A]]$
\cite{LevMan2}. By $\Psi(G)$ is denoted the set of all local maximum
independent sets of the graph $G$. Clearly, $\Omega(G)\subseteq\Psi(G)$ holds
for every graph.

\begin{theorem}
\cite{NemhTro}\label{th111} Every local maximum independent set of a graph is
a subset of a maximum independent set.
\end{theorem}

\begin{lemma}
\cite{LevMan2007b}\label{Lemma1} If $G$ is a K\"{o}nig-Egerv\'{a}ry graph,
$S\in\Psi\left(  G\right)  $, and $H=G[N[S]]$ is also a K\"{o}nig-Egerv\'{a}ry
graph, then every maximum matching of $H$ can be enlarged to a maximum
matching in $G$.
\end{lemma}

A greedoid is a set system generalizing the notion of a matroid.

\begin{definition}
\cite{BjZiegler}, \cite{KorLovSch} A \textit{greedoid} is a pair
$(E,\mathcal{F})$, where $\mathcal{F}\subseteq2^{E}$ is a non-empty set system
satisfying the following conditions:

\setlength {\parindent}{0.0cm}\textit{Accessibility}: for every non-empty
$X\in\mathcal{F}$ there is an $x\in X$ such that $X-\{x\}\in\mathcal{F}$;

\setlength {\parindent}{0.0cm}\textit{Exchange}: for $X,Y\in\mathcal{F}%
,\left\vert X\right\vert =\left\vert Y\right\vert +1$, there is an $x\in X-Y$
such that $Y\cup\{x\}\in\mathcal{F}$.
\end{definition}

It is worth observing that if $\Psi\left(  G\right)  $ is a greedoid and
$S\in$ $\Psi\left(  G\right)  $, $\left\vert S\right\vert =k\geq2$, then by
accessibility property, there is a chain%

\[
\emptyset\subset\{x_{1}\}\subset\{x_{1},x_{2}\}\subset
\text{\textperiodcentered\textperiodcentered\textperiodcentered}\subset
\{x_{1},...,x_{k-1}\}\subset\{x_{1},...,x_{k-1},x_{k}\}=S,
\]
such that $\{x_{1},x_{2},...,x_{j}\}\in\Psi\left(  G\right)  $, for all
$j\in\{1,...,k-1\}$. Such a chain is called an \textit{accessibility chain} of
$S$.

There are some known sufficient conditions on the graph $G$ that ensure the
family $\Psi(G)$ to be a greedoid on the vertex set of $G$.

\begin{theorem}
\label{ThGreed}\cite{LevMan2} \emph{(i)} If $G$ is a forest, then $\Psi(G)$ is
a greedoid.

\cite{LevMan2004} \emph{(ii)} For a bipartite graph $G,$ $\Psi(G)$ is a
greedoid if and only if all its maximum matchings are uniquely restricted.

\cite{LevMan2007b} \emph{(iii)} Let $G$ be a triangle-free graph. Then
$\Psi(G)$ is a greedoid if and only if all maximum matchings of $G$ are
uniquely restricted and the closed neighborhood of every local maximum
independent set of $G$ induces a K\"{o}nig-Egerv\'{a}ry graph.
\end{theorem}

Recall that $G$ is a \textit{well-covered} graph if all its maximal
independent sets are of the same cardinality \cite{Plummer1970}, and $G$ is
\textit{very well-covered} if, in addition, it has no isolated vertices and
$\left\vert V\left(  G\right)  \right\vert =2\alpha\left(  G\right)  $
\cite{Fav1982}.

\begin{theorem}
\label{ThVWCov}Let $G$ be a very well-covered graph. Then

\emph{(i)} \cite{LevMan2007} $G\left[  N\left[  S\right]  \right]  $ is a
K\"{o}nig-Egerv\'{a}ry graph, for every $S\in\Psi\left(  G\right)  $;

\emph{(ii)} \cite{LevMan2012c} an independent set $S$ belongs to $\Psi\left(
G\right)  $ if and only if $\ \left\vert S\right\vert =\left\vert N\left(
S\right)  \right\vert $.

\textit{(iii)} \cite{LevMan2012c} $\Psi\left(  G\right)  $ is a greedoid if
and only if $G$ has a unique maximum matching.
\end{theorem}

One of the most useful tool for the efficient optimization is the greedy
algorithm. Its essence lies in trying to find the global optimum by moving on
each step in the locally optimal direction. In particular, it resembles the
classical `gradient' search method in continuous case. If the number of steps
and the complexity of choosing the local optimum are both polynomial, then,
obviously, the greedy algorithm produces its result in polynomial time.

For some classes of problems, the greedy algorithm produces an optimal result.
For instance, Prim and Kruskal algorithms for finding a minimum spanning tree
in a graph are both greedy. In general many greedy problems can be described
using exchange structures called matroids, which were later generalized to
greedoids. One of the important tasks of combinatorial optimization theories
is to explain the correctness of the greedy approach in those cases and to
generalize them to every possible extent.

Another important task is the inverse - to define the boundaries of
applicability of the greedy algorithm. It is known that every linear objective
function can be optimized by a greedy algorithm on matroids, and conversely,
the Rado-Edmonds theorem claims that this property characterizes matroids for
hereditary set systems. In contrast with this, optimizing a linear function on
greedoids is an intractable problem. On the other hand, it was shown that
Gaussian greedoids allow a generalization of the Rado-Edmonds theorem
\cite{Bagotskaya1,BjZiegler,Goecke1,Serganova1}. In addition, there is an
algorithmic characterization of antimatroids \cite{BoydFaigle,KempLev2003}.

In what follows, we coin an idea of a more general exchange structure that we
call an augmentoid.

\begin{definition}
An \textit{augmentoid} is a pair $(E,\mathcal{F})$, where $\mathcal{F}%
\subseteq2^{E}$ is a non-empty family of sets satisfying the following condition:

\setlength {\parindent}{0.0cm}\textit{Augmentation}: for $X,Y\in\mathcal{F}$,
there exist $A\subseteq X-Y$ and $B\subseteq Y-X$ such that $Y\cup A,X\cup
B\in\mathcal{F}$ and $\left\vert Y\cup A\right\vert =\left\vert X\cup
B\right\vert $.
\end{definition}

Every element of $\mathcal{F}$ is a \textit{feasible} set of the augmentoid
$(E,\mathcal{F})$.

\begin{proposition}
\label{claim1}If $(E,\mathcal{F})$ is an augmentoid, then every feasible set
may be enlarged to a maximum feasible set.
\end{proposition}

\begin{proof}
Let $X,Y\in\mathcal{F}$, where $Y$ is a maximum feasible set. By the
augmentation property, there is a possibility to enlarge $X$ to a feasible set
$X\cup B$ in such a way that $\left\vert X\cup B\right\vert =\left\vert
Y\cup\emptyset\right\vert =\left\vert Y\right\vert $.
\end{proof}

Clearly every greedoid is an augmentoid. As we mentioned before, greedoids
were invented in order to accommodate greedy algorithms with exchange
properties. The main purpose of this paper is to extend the area of
applicability of exchange structures to a more general context including
critical sets, crowns, and local maximum independent sets.

\section{Preliminary results}

Let us notice that if $S$ is a crown of a graph $G$, having $M$ as a matching
from $N(S)$ into $S$, then $\left\vert S\right\vert +\left\vert M\right\vert
=\left\vert S\right\vert +\left\vert N(S)\right\vert =\left\vert S\cup
N(S)\right\vert $. Taking into account Theorems \ref{ThKE}, \ref{ThKEMatch},
one may immediately conclude with the following.

\begin{corollary}
\label{Cor1}If $S$ is a crown of a graph $G$, then $G\left[  N\left[
S\right]  \right]  $ is a K\"{o}nig-Egerv\'{a}ry graph.
\end{corollary}

Consequently, if $S$ is a straight crown, then $G\left[  N\left[  S\right]
\right]  $ is a K\"{o}nig-Egerv\'{a}ry graph with a perfect
matching.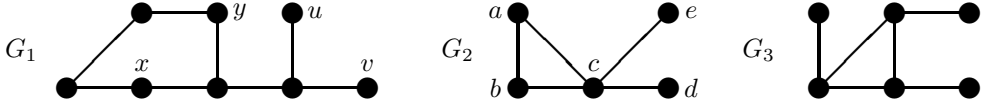
\begin{figure}[h]
\setlength{\unitlength}{1cm}\begin{picture}(5,1.4)\thicklines
\multiput(1,0)(1,0){5}{\circle*{0.29}}
\multiput(2,1)(1,0){3}{\circle*{0.29}}
\put(1,0){\line(1,1){1}}
\put(1,0){\line(1,0){4}}
\put(2,1){\line(1,0){1}}
\put(3,0){\line(0,1){1}}
\put(4,0){\line(0,1){1}}
\put(2,0.3){\makebox(0,0){$x$}}
\put(3.3,1){\makebox(0,0){$y$}}
\put(4.3,1){\makebox(0,0){$u$}}
\put(5,0.3){\makebox(0,0){$v$}}
\put(0.4,0.5){\makebox(0,0){$G_{1}$}}
\multiput(7,0)(1,0){3}{\circle*{0.29}}
\multiput(7,1)(2,0){2}{\circle*{0.29}}
\put(7,1){\line(1,-1){1}}
\put(7,0){\line(0,1){1}}
\put(7,0){\line(1,0){2}}
\put(8,0){\line(1,1){1}}
\put(6.7,0){\makebox(0,0){$b$}}
\put(6.7,1){\makebox(0,0){$a$}}
\put(8,0.3){\makebox(0,0){$c$}}
\put(9.3,0){\makebox(0,0){$d$}}
\put(9.3,1){\makebox(0,0){$e$}}
\put(6.2,0.5){\makebox(0,0){$G_{2}$}}
\multiput(11,0)(1,0){3}{\circle*{0.29}}
\multiput(11,1)(1,0){3}{\circle*{0.29}}
\put(11,0){\line(0,1){1}}
\put(11,0){\line(1,0){2}}
\put(11,0){\line(1,1){1}}
\put(12,0){\line(0,1){1}}
\put(12,1){\line(1,0){1}}
\put(10.2,0.5){\makebox(0,0){$G_{3}$}}
\end{picture}\caption{Only $G_{2}$ and $G_{3}$ are K\"{o}nig-Egerv\'{a}ry
graphs.}%
\label{Fig22}%
\end{figure}

Let us consider the graphs from Figure \ref{Fig22}. Notice that:

\begin{itemize}
\item $CritIndep(G_{1})\neq Crown(G_{1})\neq\Psi(G_{1})$,

because $\{v\}\in Crown(G_{1})-CritIndep(G_{1})$ and $\left\{  x,y\right\}
\in\Psi(G_{1})-Crown(G_{1})$;

\item $CritIndep(G_{2})\neq Crown(G_{2})\neq\Psi(G_{2})$,

since $\{d\}\in Crown(G_{2})-CritIndep(G_{2})$ and $\left\{  a\right\}
\in\Psi(G_{2})-Crown(G_{2})$;

\item $CritIndep(G_{3})=Crown(G_{3})=\Psi(G_{3})$.
\end{itemize}

\begin{theorem}
\label{ThCrown}$CritIndep(G)\subseteq Crown(G)\subseteq\Psi(G)$ hold for every
graph $G$.
\end{theorem}

\begin{proof}
Let $S\in CritIndep(G)$. By Lemma \ref{MatchLemma}, there exists a matching
from $N(S)$ into $S$. Consequently, $S$ is a crown of $G$. Hence we infer that
$CritIndep(G)\subseteq Crown(G)$.

Let now $S\in Crown(G)$. Then $G\left[  S\cup N(S)\right]  $ is a
K\"{o}nig-Egerv\'{a}ry graph, by Corollary \ref{Cor1}. Further, according to
Theorem \ref{ThKE}\textit{(ii)}, $S$ is a maximum independent in $G[N[S]]$,
which means, in other words, that $S\in$ $\Psi(G)$. Finally, we may conclude
that $Crown(G)\subseteq\Psi(G)$.
\end{proof}

By Theorem \ref{ThCrown}, each critical independent is also a crown. The
converse is not generally true. For instance, if $K_{n,1}=(A,B,E)$ has
$\left\vert A\right\vert =n\geq2$, then every proper subset $S$ of $A$ is a
crown but is not critical, because $d(S)=\left\vert S\right\vert -1<d(A)=n-1$.
Notice that $A$ is both a maximum critical independent set and a maximum crown
for $K_{n,1}$.

\begin{corollary}
\cite{LevMan2012b} Every critical independent set is a local maximum
independent set as well.
\end{corollary}

Combining Theorems \ref{th111} and \ref{ThCrown}, we deduce the following.

\begin{corollary}
\cite{ButTruk2007} Each critical independent set is included in a maximum
independent set.
\end{corollary}

\begin{corollary}
$\mathrm{\ker}(G)$ is a crown and, hence, a local maximum independent set.
\end{corollary}

Combining Theorems \ref{th5} and \ref{ThCrown}, we get the following.

\begin{corollary}
If $G$ is a K\"{o}nig-Egerv\'{a}ry graph, then:

\emph{(i)} every maximum independent set of $G$ is a crown;

\emph{(ii)} $\mathrm{core}(G)$ is a crown and, hence, a local maximum
independent set.
\end{corollary}

\begin{corollary}
\cite{AFLS2007} If $G$ is a graph with a crown $S$, then there is a vertex
cover $A$ of $G$ of minimum size that contains all the vertices in $A$ and
none of the vertices in $S$.
\end{corollary}

\begin{proof}
By Theorem \ref{ThCrown}, it follows that $S\in\Psi(G)$. Hence, by Theorem
\ref{th111}, there is some $B\in\Omega\left(  G\right)  $, such that
$S\subseteq B$. Therefore, $A\subseteq V\left(  G\right)  -B$, and this
completes the proof, because $V\left(  G\right)  -B$ is a minimum vertex cover
of $G$.
\end{proof}

Once a crown $S$ is found in a graph $G$, one can remove the vertices of
$S\cup N(S)$ and their adjacent edges, to get a smaller graph. In this way,
the problem size is now $n_{0}=n-\left\vert S\right\vert -\left\vert
N(S)\right\vert $, and the parameter size is $k_{0}=k-\left\vert H\right\vert
$.

\section{Critical sets and crowns}

\begin{lemma}
\label{lem4}If $A,B\in Crown(G)$, then there exists a perfect matching between
$A\cap N\left(  B\right)  $ and $B\cap N\left(  A\right)  $.
\end{lemma}

\begin{proof}
Consider matchings: $M_{A}$ from $N\left(  A\right)  $ into $A$ and $M_{B}$
from $N\left(  B\right)  $ into $B$. Since $B\cap N\left(  A\right)  \subseteq
N\left(  A\right)  $, the matching $M_{A}$ induces an injective mapping
$M_{1}$ from $B\cap N\left(  A\right)  $ into $A\cap N\left(  B\right)  $.
Similarly, the corresponding restriction of $M_{B}$, say $M_{2}$, maps $A\cap
N\left(  B\right)  $ into $B\cap N\left(  A\right)  $. Thus, it follows that
$\left\vert B\cap N\left(  A\right)  \right\vert =\left\vert A\cap N\left(
B\right)  \right\vert $, and, finally, both $M_{1}$ and $M_{2}$ are bijections.
\end{proof}

Theorem \ref{ThCrown} and Lemma \ref{lem4} imply the following.

\begin{corollary}
\cite{Larson2007} If $A,B\in CritIndep(G)$, then there exists a perfect
matching between $A\cap N\left(  B\right)  $ and $B\cap N\left(  A\right)  $.
\end{corollary}

\begin{remark}
It seems amusing that the matching from Lemma \ref{lem4} exists for so-called
side critical independent sets of a bipartite graph \cite{LevMan2011b}.
\end{remark}

Notice that, according to Theorems \ref{th5} and \ref{ThKEMatch}, in a
K\"{o}nig-Egerv\'{a}ry graph $G$, every $S\in\Omega\left(  G\right)  $ is both
a maximum critical independent set and a maximum crown.

\begin{theorem}
\label{th9}$\left(  V\left(  G\right)  ,Crown(G)\right)  $ is an augmentoid
for every graph $G$.
\end{theorem}

\begin{proof}
Suppose $A$ and $B$ are crowns. Let $A^{1}=A\cup B_{0}$, where $B_{0}%
=B-N\left[  A\right]  $.

First, we show that $A^{1}$ is a crown, i.e., there is a matching $M^{1}$ from
$N\left(  A^{1}\right)  =N\left(  A\right)  \cup N\left(  B_{0}\right)  $ into
$A^{1}$. Let us define it on the set $A\subseteq A^{1}$ as the matching
$M_{A}$ from $N\left(  A\right)  $ into $A$, which exists by definition of a
crown. We are left with $N\left(  B_{0}\right)  -N\left(  A\right)  $ to
handle its part in $M^{1}$.

There are no edges between $N\left(  B_{0}\right)  -N\left(  A\right)  $ and
$A\cap B$.

The set $\left(  A\cap N\left(  B\right)  \right)  \cap\left(  N\left(
B_{0}\right)  -N\left(  A\right)  \right)  $ is empty, because $B_{0}\cup A$
is independent.

Since $B$ is a crown, there is a matching $M_{B}$ from $N\left(  B\right)  $
into $B$, which maps the set $A\cap N\left(  B\right)  $ onto the set $B\cap
N\left(  A\right)  $ (as in the proof of Lemma \ref{lem4}). Thus, the
restriction $M_{B}^{\ast}$ of $M_{B}$ to $N\left(  B_{0}\right)  -N\left(
A\right)  $, maps $N\left(  B_{0}\right)  -N\left(  A\right)  $ into $B_{0}$.
Consequently, the set $M_{A}\cup M_{B}^{\ast}$ is a matching from $N\left(
A^{1}\right)  =N\left(  A\right)  \cup N\left(  B_{0}\right)  $ into $A^{1}$.

Similarly, if $B^{1}=B\cup A_{0}$, with $A_{0}=A-N\left[  B\right]  $, then
$B^{1}$ is a crown as well.

Second, we show that\textit{ }$\left\vert A^{1}\right\vert =$ $\left\vert
B^{1}\right\vert $. Indeed,%
\begin{align*}
A^{1}  &  =A_{0}\cup\left(  A\cap N\left(  B\right)  \right)  \cup\left(
A\cap B\right)  \cup B_{0},\\
B^{1}  &  =B_{0}\cup\left(  B\cap N\left(  A\right)  \right)  \cup\left(
B\cap A\right)  \cup A_{0}.
\end{align*}
Hence, in accordance with Lemma \ref{lem4} we obtain
\begin{gather*}
\left\vert A^{1}\right\vert =\left\vert A_{0}\right\vert +\left\vert A\cap
N\left(  B\right)  \right\vert +\left\vert A\cap B\right\vert +\left\vert
B_{0}\right\vert =\\
=\left\vert B_{0}\right\vert +\left\vert B\cap N\left(  A\right)  \right\vert
+\left\vert A\cap B\right\vert +\left\vert A_{0}\right\vert =\left\vert
B^{1}\right\vert ,
\end{gather*}
as required.
\end{proof}

Now, Proposition \ref{claim1} and Theorem \ref{th9} imply the following.

\begin{corollary}
\label{cor3}Every crown may be enlarged to a maximum crown.
\end{corollary}

Actually, if, in the proof of Theorem \ref{th9}, the set $B$ happens to be a
maximum crown, we infer that $A_{0}=\emptyset$, and, consequently,
$A^{1}\subseteq N\left[  B\right]  $. This leads to the following.

\begin{corollary}
\label{cor4}Every crown is included in the closed neighborhood of each maximum crown.
\end{corollary}

Combining Corollary \ref{cor4} and Theorem \ref{ThCrown}, we obtain the following.

\begin{corollary}
Every critical set is included in the closed neighborhood of each maximum crown.
\end{corollary}

\begin{theorem}
\label{cor8}If $A,B\in MaxCrown(G)$, then both $d(A)=d(B)$ and $N[A]=N[B]$.
\end{theorem}

\begin{proof}
Let $A,B\in MaxCrown(G)$, $H_{1}=G[N[A]]$ and $H_{2}=G[N[B]]$. According to
Corollary \ref{Cor1}, both $H_{1}$ and $H_{2}$ are K\"{o}nig-Egerv\'{a}ry graphs.

By Corollary \ref{cor4}, both $A\subseteq N_{G}[B]$ and $B\subseteq N_{G}[A]$.
Hence $A\in\Omega(H_{2})$ and $B\in\Omega(H_{1})$, because $A$, $B$ are
independent and $\left\vert A\right\vert =\left\vert B\right\vert $. It
follows that%
\begin{align*}
d_{H_{1}}(A)  &  =\left\vert A\right\vert -\left\vert N_{H_{1}}\left(
A\right)  \right\vert =\left\vert A\right\vert -\left\vert N_{G}\left(
A\right)  \right\vert =d_{G}(A)=d(A)\text{ and}\\
d_{H_{2}}(B)  &  =\left\vert B\right\vert -\left\vert N_{H_{2}}\left(
B\right)  \right\vert =\left\vert B\right\vert -\left\vert N_{G}\left(
B\right)  \right\vert =d_{G}(B)=d(B).
\end{align*}

By Theorem \ref{th5}, $d_{H_{1}}(A)=d_{H_{1}}(B)$ and $d_{H_{2}}(B)=d_{H_{2}%
}(A)$, because $A$ and $B$ are maximum independent sets in both $H_{1}$ and
$H_{2}$.

Since $N_{H_{1}}\left(  B\right)  \subseteq N_{G}\left(  B\right)  $, we get
\begin{align*}
d(A)  &  =d_{G}(A)=d_{H_{1}}(A)=d_{H_{1}}(B)=\\
\left\vert B\right\vert -\left\vert N_{H_{1}}\left(  B\right)  \right\vert  &
\geq\left\vert B\right\vert -\left\vert N_{G}\left(  B\right)  \right\vert
=d_{G}(B)=d(B).
\end{align*}

Since $N_{H_{2}}\left(  A\right)  \subseteq N_{G}\left(  A\right)  $, we
obtain
\begin{align*}
d(B)  &  =d_{G}(B)=d_{H_{2}}(B)=d_{H_{2}}(A)=\\
\left\vert A\right\vert -\left\vert N_{H_{2}}\left(  A\right)  \right\vert  &
\geq\left\vert A\right\vert -\left\vert N_{G}\left(  A\right)  \right\vert
=d_{G}(A)=d(A).
\end{align*}

Consequently, we conclude that $d(A)=d(B)$.

Moreover, we infer that $\left\vert N_{H_{1}}\left(  B\right)  \right\vert
=\left\vert N_{G}\left(  B\right)  \right\vert $, which implies $N_{H_{1}%
}\left(  B\right)  =N_{G}\left(  B\right)  $. Consequently, we get that
\[
N_{G}\left[  B\right]  =B\cup N_{G}\left(  B\right)  =B\cup N_{H_{1}}\left(
B\right)  \subseteq N_{G}\left[  A\right]  .
\]

Similarly, one can deduce that $N_{G}\left[  A\right]  =A\cup N_{G}\left(
A\right)  =A\cup N_{H_{2}}\left(  A\right)  \subseteq N_{G}\left[  B\right]  $.

Therefore, $N_{G}[A]=N_{G}[B]$.
\end{proof}

\begin{lemma}
\label{lem3}If $A\subseteq B$ and $B$ is a crown, then $d\left(  A\right)
\leq d\left(  B\right)  $.
\end{lemma}

\begin{proof}
Since $B$ is a crown, there is a matching $M_{B}$ from $N\left(  B\right)  $
into $B$. The matching $M_{B}$ maps $N\left(  B\right)  -N\left(  A\right)  $
into $B-A$. Therefore,
\begin{gather*}
d\left(  B\right)  =\left\vert B\right\vert -\left\vert N\left(  B\right)
\right\vert =\left\vert A\right\vert +\left\vert B-A\right\vert -\left(
\left\vert N\left(  A\right)  \right\vert +\left\vert N\left(  B\right)
-N\left(  A\right)  \right\vert \right)  =\\
=\left\vert A\right\vert -\left\vert N\left(  A\right)  \right\vert
+\left\vert B-A\right\vert -\left\vert N\left(  B\right)  -N\left(  A\right)
\right\vert \geq\left\vert A\right\vert -\left\vert N\left(  A\right)
\right\vert =d(A),
\end{gather*}
as required.
\end{proof}

\begin{theorem}
\label{th11}$\left(  V\left(  G\right)  ,CritIndep(G)\right)  $ is an
augmentoid for every graph $G$.
\end{theorem}

\begin{proof}
Suppose $A$ and $B$ are critical independent sets. Hence they are crowns. By
Theorem \ref{th9}, $A^{1}=A\cup\left(  B-N\left[  A\right]  \right)  $ and
$B^{1}=B\cup\left(  A-N\left[  B\right]  \right)  $ are crowns as well and
$\left\vert A^{1}\right\vert =\left\vert B^{1}\right\vert $. Since $A\subseteq
A^{1}$, Lemma \ref{lem3} implies that $d\left(  A\right)  \leq d\left(
A^{1}\right)  $. On the other hand, $d\left(  A^{1}\right)  \leq d\left(
A\right)  $, because $A$ is critical. Thus $A^{1},B^{1}\in$ $CritIndep(G)$.
\end{proof}

Now, Proposition \ref{claim1} and Theorem \ref{th11} imply the following.

\begin{corollary}
\cite{Larson2011} An inclusion maximal critical independent set is a maximum
critical independent set.
\end{corollary}

\begin{theorem}
\label{th10}$MaxCritIndep(G)=MaxCrown(G)$.
\end{theorem}

\begin{proof}
Let $A\in MaxCritIndep(G)$ and $B\in MaxCrown(G)$. By Theorem \ref{ThCrown},
we have that $A\in Crown(G)$.

According to Theorem \ref{th9}, there exists some $C\in MaxCrown(G)$ such that
$A\subseteq C$.

By Lemma \ref{lem3} and Theorem \ref{cor8}, we get
\[
d(G)=d\left(  A\right)  \leq d\left(  C\right)  =d\left(  B\right)  \leq
d(G)\text{,}%
\]
which implies $d\left(  C\right)  =d\left(  B\right)  =d(G)$, i.e., $C,B\in
CritIndep(G)$.

Since $A\subseteq C$ and $A\in MaxCritIndep(G)$, we get that $A=C\in$
$MaxCritIndep(G)$. Thus $A\in MaxCrown(G)$.

Since $C\in$ $MaxCritIndep(G)$, $B\in CritIndep(G)$ and $\left\vert
B\right\vert =\left\vert C\right\vert $, it follows that $B\in
MaxCritIndep(G)$.
\end{proof}

Combining Theorem \ref{ThCrown}, Corollary \ref{cor3}, and Theorem \ref{th10},
we get the following.

\begin{corollary}
\cite{Larson2007} Every critical independent set is contained in a maximum
critical independent set.
\end{corollary}

\begin{lemma}
\label{Lemma2}If $A$ $\in CritIndep(G)$, then $\alpha(G)=\alpha\left(
G[N\left[  A\right]  ]\right)  +\alpha(G[V(G)-N[A]])$.
\end{lemma}

\begin{proof}
By Theorem \ref{ThCrown}, $A$ $\in Crown(G)\subseteq\Psi(G)$, which means that
$G[N\left[  A\right]  ]$ is a K\"{o}nig-Egerv\'{a}ry graph with $\alpha\left(
G[N\left[  A\right]  ]\right)  =\left\vert A\right\vert $. Theorem \ref{th111}
implies that there exists $S\in\Omega(G)$ such that $A\subseteq S$. Since no
edge joins any $a\in A$ to vertices belonging to $S-A$ or to $N(S)-N(A)$, it
follows that $\alpha(G[V(G)-N[A])=\alpha(G)-\left\vert A\right\vert $.
\end{proof}

\begin{corollary}
\cite{Larson2011} For any graph $G$, there is a unique set $X\subseteq V(G)$
such that

\emph{(i)} $\alpha(G)=\alpha(G[X])+\alpha(G[V(G)-X])$;

\emph{(ii)} for every maximum critical independent set $A$ of $G$, $X$
$=N\left[  A\right]  $;

\emph{(iii)} $G[X]$ is a K\"{o}nig-Egerv\'{a}ry graph.
\end{corollary}

\begin{proof}
Let $A,B\in MaxCritIndep(G)$ and $X$ $=N\left[  B\right]  $.

By Lemma \ref{Lemma2}, we get Part \textit{(i)}.

By Theorem \ref{th10}, $A,B\in MaxCrown(G)$, and Theorem \ref{cor8} implies
$N\left[  A\right]  =N\left[  B\right]  $, i.e., Part \emph{(ii)} holds.

Since $B\in MaxCrown(G)\subseteq Crown(G)$, we have that $G[X]=G[N[B]]$ is a
K\"{o}nig-Egerv\'{a}ry graph, which completes the proof.
\end{proof}

Butenko and Trukhanov showed that identifying a non-empty critical independent
set gives a polynomial-time reduction of the problem of finding a maximum
independent set to a proper subgraph. Since a maximum critical independent set
may be included in a maximum independent set of a graph, both it and its
neighbors may be removed, reducing the problem of finding a maximum
independent set to a subgraph. Thus critical independent sets give a
polynomially tractable step for the well-known \textbf{NP}-hard problem of
finding a maximum independent set in a graph \cite{GaryJohnson79}. Moreover,
there are parallel algorithms returning maximum critical independent sets
\cite{DeLaVinaLarson2013}.

\begin{theorem}
\label{th22}\cite{Larson2007} One can compute a maximum critical independent
set in polynomial time.
\end{theorem}

Combining Theorems \ref{th10} and \ref{th22}, we get the following.

\begin{corollary}
\cite{AFLS2007} Isolating a crown of maximum order is solvable in polynomial time.
\end{corollary}

If $\mathrm{core}(G)$ is not critical set, then $\mathrm{core}(G)$ is not
necessarily a crown. For instance, consider the graph from Figure \ref{Fig1},
whose $\mathrm{core}(G)$ is neither a critical set, nor a crown.

\begin{figure}[h]
\setlength{\unitlength}{1cm}\begin{picture}(5,1.2)\thicklines
\multiput(4,0)(1,0){6}{\circle*{0.29}}
\multiput(5,1)(1,0){3}{\circle*{0.29}}
\put(9,1){\circle*{0.29}}
\put(4,0){\line(1,0){5}}
\put(5,0){\line(0,1){1}}
\put(6,0){\line(0,1){1}}
\put(6,1){\line(1,0){1}}
\put(7,1){\line(1,-1){1}}
\put(8,0){\line(1,1){1}}
\put(9,0){\line(0,1){1}}
\put(3.7,0){\makebox(0,0){$a$}}
\put(4.7,1){\makebox(0,0){$b$}}
\put(7,0.3){\makebox(0,0){$c$}}
\put(3,0.5){\makebox(0,0){$G$}}
\end{picture}\caption{$\mathrm{core}(G)=\{a,b,c\}$ and $d(core(G))=0<1=d(G)$.}%
\label{Fig1}%
\end{figure}
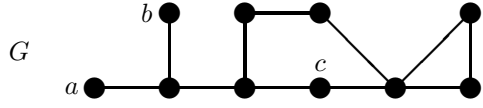

\begin{remark}
If $A,B\in Crown(G)$, then $A\cap B$ is not necessarily a crown; e.g.,
$P_{5}=\left(  \{v_{1},v_{2},v_{3},v_{4},v_{5}\},\{v_{1}v_{2},v_{2}v_{3}%
,v_{3}v_{4},v_{4}v_{5}\}\right)  $ has $\{v_{1},v_{3}\},\{v_{3},v_{5}\}\in
$.$Crown(P_{5})$, while $\{v_{1},v_{3}\}\cap\{v_{3},v_{5}\}=\{v_{3}\}\notin
Crown(P_{5})$.
\end{remark}

\begin{proposition}
If $A,B\in Crown(G)$, such that $A\cup B$ is independent, then $A\cup B\in
Crown(G)$ as well.
\end{proposition}

\begin{proof}
There exist a matching $M_{1}$ from $N(A)$ into $A$ and a matching $M_{2}$
from $N(B)$ into $B$. Let $M_{3}$ be the restriction of $M_{1}$ to
$N(A)-N(B)$. Then, $M_{2}\cup M_{3}$ is a matching from $N(A)\cup N(B)$ into
$A\cup B$. Consequently, $A\cup B$ is a crown.
\end{proof}

It is worth mentioning that a similar result is true for local maximum
independent sets.

\begin{proposition}
\cite{LevMan2} If $A$ and $B$ are two disjoint local maximum independent sets
in $G$, such that $A\cup B$ is independent, then $A\cup B$ is also a local
maximum independent set in $G$.
\end{proposition}

\section{Crowns and local maximum independent sets}

Notice that there exists a non-bipartite K\"{o}nig-Egerv\'{a}ry graph $G$ with
$Crown(G)\neq\Psi(G)$. For instance, the graph $G_{1}$ from Figure \ref{Fig12}
is a non-bipartite K\"{o}nig-Egerv\'{a}ry graph, $\left\{  d\right\}  \in
Crown(G_{1})\cap\Psi(G_{1})$, while $\left\{  a\right\}  \in\Psi
(G_{1})-Crown(G_{1})$.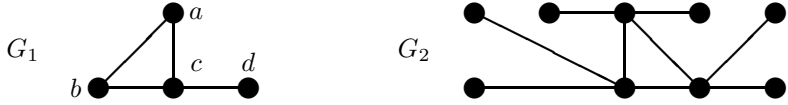
\begin{figure}[h]
\setlength{\unitlength}{1cm}\begin{picture}(5,1.4)\thicklines
\multiput(3,0)(1,0){3}{\circle*{0.29}}
\put(4,1){\circle*{0.29}}
\put(3,0){\line(1,1){1}}
\put(4,0){\line(0,1){1}}
\put(3,0){\line(1,0){2}}
\put(2.7,0){\makebox(0,0){$b$}}
\put(4.3,1){\makebox(0,0){$a$}}
\put(4.3,0.3){\makebox(0,0){$c$}}
\put(5,0.35){\makebox(0,0){$d$}}
\put(2,0.5){\makebox(0,0){$G_{1}$}}
\multiput(8,0)(0,1){2}{\circle*{0.29}}
\multiput(9,1)(1,0){4}{\circle*{0.29}}
\multiput(10,0)(1,0){3}{\circle*{0.29}}
\put(8,0){\line(1,0){4}}
\put(8,1){\line(2,-1){2}}
\put(9,1){\line(1,0){2}}
\put(10,1){\line(1,-1){1}}
\put(10,0){\line(0,1){1}}
\put(11,0){\line(1,1){1}}
\put(7.2,0.5){\makebox(0,0){$G_{2}$}}
\end{picture}\caption{$Crit(G_{1})\neq Crown(G_{1})$ and $Crown(G_{2}%
)=\Psi(G_{2})$.}%
\label{Fig12}%
\end{figure}

\begin{proposition}
\label{Prop2}$Crown(G)=\Psi(G)$ if and only if $G[N[S]]$ is a
K\"{o}nig-Egerv\'{a}ry graph for every $S\in\Psi\left(  G\right)  $. In
particular, $G$ is a K\"{o}nig-Egerv\'{a}ry graph.
\end{proposition}

\begin{proof}
By Corollary \ref{Cor1}, $G[N[S]]$ is a K\"{o}nig-Egerv\'{a}ry graph for every
$S\in\Psi\left(  G\right)  $, because $\Psi(G)=Crown(G)$.

Conversely, by Theorem \ref{ThKE} for each $S\in\Psi\left(  G\right)  $ there
is a matching from $N\left(  S\right)  $ to $S$, i.e. $S\in Crown\left(
G\right)  $.
\end{proof}

Since every subgraph of a bipartite graph is bipartite
(K\"{o}nig-Egerv\'{a}ry), we conclude with the following.

\begin{corollary}
\cite{LevMan2016}\label{ThCrownBip} If $G$ is a bipartite graph, then
$Crown(G)=\Psi(G)$.
\end{corollary}

The converse of Corollary \ref{ThCrownBip} is not generally true. For
instance, the non-bipartite K\"{o}nig-Egerv\'{a}ry graph $G_{2}$ from Figure
\ref{Fig12} satisfies $Crown(G_{2})=\Psi(G_{2})$.

Combining Theorems \ref{ThGreed}\textit{(i),(ii)} and \ref{ThCrownBip}, we get
the following.

\begin{corollary}
\emph{(i)} For every tree, $Crown(T)$ is a greedoid.

\emph{(ii)} If $G$ is a bipartite graph with a unique maximum matching, then
$Crown(G)$ is a greedoid.
\end{corollary}

The graph $G_{1}=K_{3}\circ K_{1}$ from Figure \ref{fig33} is not bipartite,
and $Crown(G)=\Psi(G)$. Notice that $K_{3}\circ K_{1}$ is very well-covered.

\begin{figure}[h]
\setlength{\unitlength}{1cm}\begin{picture}(5,1.1)\thicklines
\multiput(3,0)(1,0){4}{\circle*{0.29}}
\multiput(4,1)(1,0){2}{\circle*{0.29}}
\put(3,0){\line(1,0){3}}
\put(4,0){\line(0,1){1}}
\put(4,1){\line(1,0){1}}
\put(4,1){\line(1,-1){1}}
\put(2,0.5){\makebox(0,0){$G_{1}$}}
\multiput(8,0)(1,0){4}{\circle*{0.29}}
\multiput(9,1)(1,0){3}{\circle*{0.29}}
\put(8,0){\line(1,0){3}}
\put(9,1){\line(1,-1){1}}
\put(9,1){\line(1,0){1}}
\put(10,0){\line(0,1){1}}
\put(11,0){\line(0,1){1}}
\put(8.7,1){\makebox(0,0){$x$}}
\put(7,0.5){\makebox(0,0){$G_{2}$}}
\end{picture}\caption{Both $G_{1}$ and $G_{2}$ are well-covered.}%
\label{fig33}%
\end{figure}
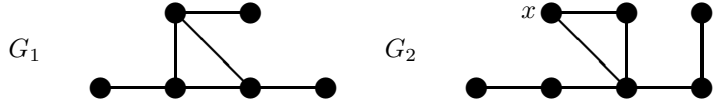

\begin{proposition}
\label{Prop7}If $G$ is a very well-covered graph, then $Crown(G)=\Psi(G)$.
\end{proposition}

\begin{proof}
By Theorem \ref{ThCrown}, it is enough to show that $\Psi(G)\subseteq
Crown(G)$. If $S\in\Psi\left(  G\right)  $, then $S$ is independent and
$G\left[  N\left[  S\right]  \right]  $ is a K\"{o}nig-Egerv\'{a}ry graph, by
Theorem \ref{ThVWCov}\emph{(i)}. Now, the conclusion follows from Proposition
\ref{Prop2}.
\end{proof}

Since $H\circ K_{1}$ is very well-covered, for every graph $H$, we infer the following.

\begin{corollary}
If $G=H\circ K_{1}$, then $Crown(G)=\Psi(G)$.
\end{corollary}

Proposition \ref{Prop7} is not true for all well-covered graphs. For instance,
the graph $G_{2}$ from Figure \ref{fig33} is well-covered, and $\left\{
x\right\}  \in\Psi(G)-Crown(G)$.

Taking into account Theorem \ref{ThVWCov}\textit{(iii) }and
Proposition\textit{ }\ref{Prop7}, we get the following.

\begin{corollary}
If $G$ is a very well-covered graph with a unique maximum matching, then
$Crown(G)$ is a greedoid.
\end{corollary}

There are non-bipartite and non-very well-covered graphs satisfying
$Crown(G)\neq\Psi(G)$; e.g., the graph $G_{2}$ from Figure \ref{Fig12}, where
$\{x\}\in\Psi(G)-Crown(G)$.

If $H$ is a K\"{o}nig-Egerv\'{a}ry subgraph of a graph $G$, it is not true
that there is some maximum independent set $S$ in $H$, such that both $S$ is a
crown of $G$ and $G\left[  S\cup N_{G}(S)\right]  =H$. See, for instance, the
subgraphs $H_{1}$ and $H_{2}$ of the graphs $G_{1}$ and $G_{2}$, respectively,
from Figure \ref{fig1444} (notice that $G_{1}$ is bipartite, while $G_{2}$ is
a non-bipartite K\"{o}nig-Egerv\'{a}ry graph).

\begin{figure}[h]
\setlength{\unitlength}{1cm}\begin{picture}(5,1.4)\thicklines
\multiput(2,0)(1,0){5}{\circle*{0.29}}
\multiput(3,1)(1,0){4}{\circle*{0.29}}
\put(2,0){\line(1,1){1}}
\put(3,0){\line(0,1){1}}
\put(3,0){\line(1,1){1}}
\put(4,0){\line(0,1){1}}
\put(4,1){\line(1,0){1}}
\put(5,0){\line(0,1){1}}
\put(5,0){\line(1,1){1}}
\put(5,1){\line(1,-1){1}}
\put(6,0){\line(0,1){1}}
\put(2.7,0){\makebox(0,0){$x$}}
\put(3.7,0){\makebox(0,0){$y$}}
\put(4.7,0){\makebox(0,0){$z$}}
\put(3,1.35){\makebox(0,0){$u$}}
\put(4,1.35){\makebox(0,0){$v$}}
\put(5,1.35){\makebox(0,0){$w$}}
\put(1.2,0.5){\makebox(0,0){$G_{1}$}}
\multiput(8,0)(1,0){5}{\circle*{0.29}}
\multiput(9,1)(1,0){4}{\circle*{0.29}}
\put(8,0){\line(1,1){1}}
\put(9,0){\line(0,1){1}}
\put(9,0){\line(1,1){1}}
\put(9,1){\line(1,0){3}}
\put(10,0){\line(0,1){1}}
\put(11,0){\line(0,1){1}}
\put(11,0){\line(1,1){1}}
\put(12,0){\line(0,1){1}}
\put(8.7,0){\makebox(0,0){$a$}}
\put(9.7,0){\makebox(0,0){$b$}}
\put(10.7,0){\makebox(0,0){$c$}}
\put(9,1.35){\makebox(0,0){$e$}}
\put(10,1.35){\makebox(0,0){$f$}}
\put(11,1.35){\makebox(0,0){$g$}}
\put(7.2,0.5){\makebox(0,0){$G_{2}$}}
\end{picture}\caption{$H_{1}=G_{1}\left[  \left\{  x,y,z,u,v,w\right\}
\right]  $ and $H_{2}=G_{2}\left[  \left\{  a,b,c,e,f,g\right\}  \right]  $
are K\"{o}nig-Egerv\'{a}ry graphs.}%
\label{fig1444}%
\end{figure}
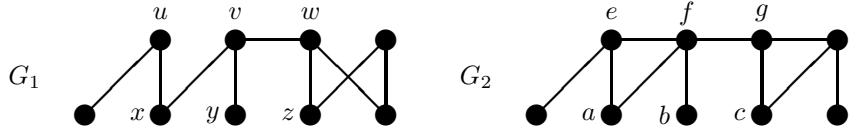

The definition of a crown and Corollary \ref{Cor1} imply the following.

\begin{corollary}
Let $H=S\ast A$ be a K\"{o}nig-Egerv\'{a}ry subgraph of $G$. Then $S\in
Crown(G)$ if and only if $N_{G}(S)=A$.
\end{corollary}

Combining Theorems \ref{ThKE}, \ref{th5}, \ref{th10} we obtain the following.

\begin{theorem}
These assertions are equivalent:

\emph{(i)} $G$ is a K\"{o}nig-Egerv\'{a}ry graph;

\emph{(ii)} $Crown(G)\cap\Omega(G)\neq\emptyset$;

\emph{(iii)} $MaxCrown(G)\cap\Omega(G)\neq\emptyset$;

\emph{(iv)} $MaxCrown(G)=\Omega(G)$.
\end{theorem}

Let notice that there exists a tree $G$, such that $CritIndep(G)\neq
Crown(G)$. For instance, the path $P_{3}=\left(  \left\{  a,b,c\right\}
,\left\{  ab,bc\right\}  \right)  $ has $CritIndep(P_{3})=\left\{  \left\{
a,c\right\}  \right\}  $, while $Crown(P_{3})=\left\{  \emptyset,\left\{
a\right\}  ,\left\{  c\right\}  ,\left\{  a,c\right\}  \right\}  $.

It is easy to see that $CritIndep(C_{2n+1})=Crown(C_{2n+1})=\{\emptyset
\}\neq\Psi(C_{2n+1})$, while $CritIndep(C_{2n})=Crown(C_{2n})=\Psi(C_{2n})$.

\begin{proposition}
\label{Prop6}If $G$ is a K\"{o}nig-Egerv\'{a}ry graph with a perfect matching,
then $CritIndep(G)=Crown(G)$.
\end{proposition}

\begin{proof}
By Theorem \ref{th8}, we get $d(G)=\alpha\left(  G\right)  -\mu\left(
G\right)  =0$. According to Theorem \ref{ThCrown}, we have
$CritIndep(G)\subseteq Crown(G)$. Suppose, to the contrary, that there is some
$A\in Crown(G)$ such that $A\notin CritIndep(G)$. It follows that
$d(A)=\left\vert A\right\vert -\left\vert N\left(  A\right)  \right\vert <0$.
Consequently, there is no matching from $N\left(  A\right)  $ into $A$, in
contradiction with $A\in Crown(G)$.
\end{proof}

The converse to Proposition \ref{Prop6} is not true. For instance,
$CritIndep(C_{5})=Crown(C_{5})$, while $C_{5}$ is not a K\"{o}nig-Egerv\'{a}ry graph.

The K\"{o}nig-Egerv\'{a}ry graphs $G_{1}$, $G_{2}$ and $G_{2}$ from Figure
\ref{fig3} have $d(G_{1})=d(G_{2})=1$, while $d(G_{3})=0$. In addition,
$\left\{  a\right\}  \in\Psi\left(  G_{1}\right)  -CritIndep(G_{1})$, and
$\left\{  u,v\right\}  \in$ $\Psi\left(  G_{2}\right)  -CritIndep(G_{2})$,
which means that both $CritIndep(G_{1})\neq\Psi(G_{1})$ and $CritIndep(G_{2}%
)\neq\Psi(G_{2})$.

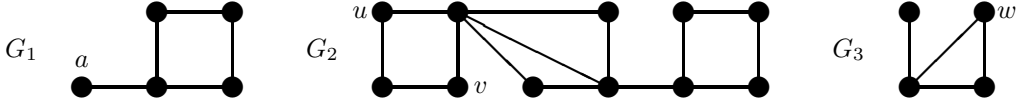
\begin{figure}[h]
\setlength{\unitlength}{1cm}\begin{picture}(5,1.4)\thicklines
\multiput(1,0)(1,0){3}{\circle*{0.29}}
\multiput(2,1)(1,0){2}{\circle*{0.29}}
\put(1,0){\line(1,0){2}}
\put(2,0){\line(0,1){1}}
\put(2,1){\line(1,0){1}}
\put(3,0){\line(0,1){1}}
\put(1,0.35){\makebox(0,0){$a$}}
\put(0.2,0.5){\makebox(0,0){$G_{1}$}}
\multiput(5,0)(1,0){6}{\circle*{0.29}}
\multiput(5,1)(1,0){2}{\circle*{0.29}}
\multiput(8,1)(1,0){3}{\circle*{0.29}}
\put(5,0){\line(1,0){1}}
\put(5,1){\line(1,0){3}}
\put(5,0){\line(0,1){1}}
\put(6,0){\line(0,1){1}}
\put(6,1){\line(1,-1){1}}
\put(6,1){\line(2,-1){2}}
\put(7,0){\line(1,0){3}}
\put(9,1){\line(1,0){1}}
\put(8,0){\line(0,1){1}}
\put(9,0){\line(0,1){1}}
\put(9,1){\line(1,0){1}}
\put(10,0){\line(0,1){1}}
\put(4.7,1){\makebox(0,0){$u$}}
\put(6.3,0){\makebox(0,0){$v$}}
\put(4.2,0.5){\makebox(0,0){$G_{2}$}}
\multiput(12,0)(1,0){2}{\circle*{0.29}}
\multiput(12,1)(1,0){2}{\circle*{0.29}}
\put(12,0){\line(1,0){1}}
\put(12,0){\line(0,1){1}}
\put(12,0){\line(1,1){1}}
\put(13,0){\line(0,1){1}}
\put(13.3,1){\makebox(0,0){$w$}}
\put(11.2,0.5){\makebox(0,0){$G_{3}$}}
\end{picture}\caption{K\"{o}nig-Egerv\'{a}ry graphs. Only $G_{3}$ has a
perfect matching.}%
\label{fig3}%
\end{figure}

\begin{proposition}
\label{lem2}If $CritIndep(G)=\Psi(G)$, then $G$ is a K\"{o}nig-Egerv\'{a}ry
graph with a perfect matching.
\end{proposition}

\begin{proof}
Since $\Omega\left(  G\right)  \subseteq\Psi\left(  G\right)  $, we infer that
every maximum independent set of $G$ is critical. Consequently, by Theorem
\ref{th5}, $G$ is a K\"{o}nig-Egerv\'{a}ry graph. Further, Theorem \ref{th8}
implies
\[
d\left(  S\right)  =d\left(  G\right)  =\alpha\left(  G\right)  -\mu\left(
G\right)  =d\left(  \emptyset\right)  =0,
\]
for every $S\in\Omega\left(  G\right)  $, because $\emptyset\in\Psi
(G)=CritIndep(G)$. Hence we obtain
\[
\alpha\left(  G\right)  =\left\vert S\right\vert =\left\vert N(S)\right\vert
=\left\vert V\left(  G\right)  -S\right\vert =\mu\left(  G\right)  ,
\]
i.e., $G$ has a perfect matching.
\end{proof}

Notice that the converse of Proposition \ref{lem2} is not true, even if $G$
has a unique perfect matching; e.g., consider the graph $G_{3}$ from Figure
\ref{fig3}, where $\left\{  w\right\}  \in\Psi(G_{3})-CritIndep(G_{3})$.

Combining Propositions \ref{lem2}, \ref{Prop6} and Theorem \ref{ThCrownBip} we
conclude with the following.

\begin{corollary}
\label{Prop1}Let $G$ be a bipartite graph. Then $CritIndep(G)=\Psi\left(
G\right)  $ if and only if $G$ has a perfect matching.
\end{corollary}

Corollary \ref{Prop1} can not be extended to K\"{o}nig-Egerv\'{a}ry graphs.
For example, consider the graph $G_{3}$ from Figure \ref{fig3}.

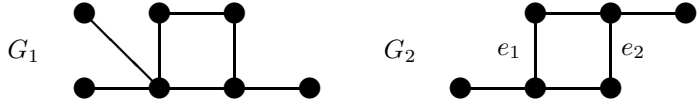
\begin{figure}[h]
\setlength{\unitlength}{1cm}\begin{picture}(5,1.1)\thicklines
\multiput(3,0)(1,0){4}{\circle*{0.29}}
\multiput(3,1)(1,0){3}{\circle*{0.29}}
\put(3,0){\line(1,0){3}}
\put(3,1){\line(1,-1){1}}
\put(4,1){\line(1,0){1}}
\put(4,0){\line(0,1){1}}
\put(5,0){\line(0,1){1}}
\put(2.2,0.5){\makebox(0,0){$G_{1}$}}
\multiput(8,0)(1,0){3}{\circle*{0.29}}
\multiput(9,1)(1,0){3}{\circle*{0.29}}
\put(8,0){\line(1,0){2}}
\put(9,0){\line(0,1){1}}
\put(10,0){\line(0,1){1}}
\put(9,1){\line(1,0){2}}
\put(8.65,0.5){\makebox(0,0){$e_{1}$}}
\put(10.3,0.5){\makebox(0,0){$e_{2}$}}
\put(7.2,0.5){\makebox(0,0){$G_{2}$}}
\end{picture}\caption{Bipartite graphs with maximum matchings that are
uniquely restricted.}%
\label{fig22}%
\end{figure}

Consider the bipartite graphs from Figure \ref{fig22}. According to Theorem
\ref{ThGreed}\textit{(ii)}, the family $\Psi(G_{1})$ is a greedoid, while
$\Psi(G_{2})$ is not a greedoid, because the maximum matching $\left\{
e_{1},e_{2}\right\}  $ is not uniquely restricted. However, $G_{1}$ has no
perfect matchings, and therefore, $CritIndep(G_{1})\neq\Psi\left(
G_{1}\right)  $, by Corollary \ref{Prop1}.

\begin{theorem}
\label{th3}Let $G$ be a bipartite graph. Then $CritIndep(G)$ is a greedoid if
and only if $G$ has a unique perfect matching.
\end{theorem}

\begin{proof}
By Theorems \ref{th5} and \ref{ThCrownBip}, we infer that
\[
\Omega\left(  G\right)  \subseteq CritIndep(G)\subseteq Crown(G)=\Psi(G).
\]
Hence, every $S\in\Omega\left(  G\right)  $ has an accessibility chain, i.e.,
there are $S_{k}\in CritIndep(G)$ such that
\[
\emptyset=S_{0}\subset S_{1}\subset...\subset S_{\alpha(G)}=S\text{ and
}\left\vert S_{k}\right\vert =k.
\]
It follows that $\left\vert S_{0}\right\vert =0=\left\vert N\left(
S_{0}\right)  \right\vert $, which, by Theorem \ref{th8}, implies
\[
d(G)=\alpha\left(  G\right)  -\mu\left(  G\right)  =d(S_{0})=0.
\]
Consequently, $G$ has a perfect matching, say $M$. Hence, Corollary
\ref{Prop1} implies $CritIndep(G)=\Psi\left(  G\right)  $. Further, according
to Theorem \ref{ThGreed}\textit{(ii)}, all maximum matchings of $G$ must be
uniquely restricted. In other words, $M$ is the unique perfect matching of $G$.

The converse follows by combining Theorem \ref{ThGreed}\textit{(ii)} and
Corollary \ref{Prop1}.
\end{proof}

Combining Theorems \ref{ThGreed}\textit{(ii)} and \ref{th3}, we obtain the following.

\begin{corollary}
\label{cor2}If $G$ is a bipartite graph such that $CritIndep(G)$ is a
greedoid, then $CritIndep(G)=\Psi\left(  G\right)  $.
\end{corollary}

The converse of Corollary \ref{cor2} is not true; e.g., consider $C_{4}$.

Let us notice that $K_{1}$ is a tree and $CritIndep(K_{1})\neq\Psi\left(
K_{1}\right)  $, since $\emptyset\in\Psi\left(  K_{1}\right)  -CritIndep(K_{1}%
)$. Taking into account that a tree may have at most one perfect matching, we
get the following.

\begin{corollary}
Let $T$ be a tree of order at least two. Then the following are equivalent:

\emph{(i)} $CritIndep(T)=Crown(G)$;

\emph{(ii)} $CritIndep(T)=\Psi\left(  T\right)  $;

\emph{(iii)} $d(T)=0$;

\emph{(iv)} $T$ has a perfect matching;

\emph{(v)} $CritIndep(T)$ is a greedoid.
\end{corollary}

\begin{proof}
\emph{(i)} $\Leftrightarrow$ \emph{(ii)} It follows by combining Theorems
\ref{ThCrown}, \ref{ThCrownBip}.

\emph{(ii)} $\Rightarrow$ \emph{(iii)} The tree $T$ has a leaf, say $v$,
because its order is at least two. Since, clearly, $\left\{  v\right\}
\in\Psi\left(  T\right)  $, we get that $\left\{  v\right\}  \in
CritIndep(T)$, and hence, $d(T)=d\left(  \left\{  v\right\}  \right)  =0$.

\emph{(iii)} $\Leftrightarrow$ \emph{(iv)} As a K\"{o}nig-Egerv\'{a}ry graph,
$T$ satisfies $d\left(  T\right)  =\alpha\left(  T\right)  -\mu\left(
T\right)  $, according to Theorem \ref{th8}. Consequently, we get that%
\[
d(T)=0\Leftrightarrow\alpha\left(  T\right)  =\mu\left(  T\right)
\Leftrightarrow T\text{ has a perfect matching}.
\]

\emph{(iii)} $\Rightarrow$ \emph{(ii)} By Theorem \ref{ThCrown}, it is enough
to show that $\Psi(T)\subseteq CritIndep(T)$.

Assume, to the contrary, that there is some $S\in$ $\Psi(T)-CritIndep(T)$.
Hence, $d\left(  S\right)  =\left\vert S\right\vert -\left\vert N\left(
S\right)  \right\vert <0=d(T)$, i.e., $\left\vert S\right\vert <\left\vert
N\left(  S\right)  \right\vert $. On the other hand, since $T[S]$ is a forest
and $S\in\Omega\left(  T[S]\right)  $, we infer that $\left\vert S\right\vert
\geq\left\vert N\left(  S\right)  \right\vert $, a contradiction.
Consequently, $CritIndep(T)=\Psi(T)$.

\emph{(ii)} $\Rightarrow$ \emph{(v)} Clear, as $\Psi\left(  T\right)  $ is a
greedoid for every tree $T$.

\emph{(v)} $\Rightarrow$ \emph{(iii)} By Theorem \ref{th5}, $\Omega\left(
T\right)  \subseteq$ $CritIndep(T)$. Hence, every $S\in\Omega\left(  T\right)
$ has an accessibility chain, i.e., there are $S_{k}\in CritIndep(T)$ such
that $S_{0}\subset S_{1}\subset...\subset S_{\alpha(T)}=S$ and $\left\vert
S_{k}\right\vert =k$. It follows that $\left\vert S_{0}\right\vert
=0=\left\vert N\left(  S_{0}\right)  \right\vert $, which implies
$d(T)=d(S_{0})=0$.
\end{proof}

Corollary \ref{Prop1} and Corollary \ref{ThCrownBip} imply that if $G$ is a
bipartite graph, then $CritIndep(G)=Crown(G)=\Psi(G)$ if and only $G$ has a
perfect matching. For instance, $CritIndep(C_{6})=\Omega\left(  C_{6}\right)
\cup\left\{  \emptyset\right\}  =\Psi(C_{6})$, while $CritIndep(C_{5}%
)=\left\{  \emptyset\right\}  \neq\Omega\left(  C_{5}\right)  \cup\left\{
\emptyset\right\}  =\Psi(C_{5})$.

\begin{theorem}
\label{th1111}$CritIndep(G)=Crown(G)=\Psi(G)$ if and only if $G[N[S]]$ is a
K\"{o}nig-Egerv\'{a}ry graph with a perfect matching for every $S\in
\Psi\left(  G\right)  $. In particular, $G$ is a K\"{o}nig-Egerv\'{a}ry graph
with a perfect matching.
\end{theorem}

\begin{proof}
Assume that $CritIndep(G)=Crown(G)=\Psi(G)$.

By Proposition \ref{lem2}, $G$ is a K\"{o}nig-Egerv\'{a}ry graph with a
perfect matching, say $M$.

By Theorem \ref{ThKEMatch}, $M\subseteq(A,V(G)-A)$, for every $A\in\Omega(G)$.

Let $S\in\Psi\left(  G\right)  $.

$\bullet$ By Proposition \ref{Prop2}, $G[N[S]]$ is a K\"{o}nig-Egerv\'{a}ry
graph. Therefore, by Theorem \ref{ThKE}\emph{(ii)}, there is a matching from
$N\left(  S\right)  $ into $S$. Hence, $\left\vert N\left(  S\right)
\right\vert \leq\left\vert S\right\vert $.

$\bullet$ By Theorem \ref{th111}, there is some $B\in\Omega(G)$ such that
$S\subseteq B$. Consequently, the trace of $M$ on $G[N[S]]$ is a matching from
$S$ into $N\left(  S\right)  $. Thus $\left\vert S\right\vert \leq\left\vert
N\left(  S\right)  \right\vert $.

$\bullet$ All in all, $\left\vert S\right\vert =\left\vert N\left(  S\right)
\right\vert $. In other words, the trace of $M$ on $G[N[S]]$ must be a perfect matching.

Conversely, by Proposition \ref{Prop6}, we have that $CritIndep(G)=Crown(G)$,
and by Proposition \ref{Prop2}, we get that $Crown(G)=\Psi(G)$.
\end{proof}

For example, the graph $G$ from Figure \ref{fig1211} is a
K\"{o}nig-Egerv\'{a}ry graph, $S=\{a,b,c,d\}\in\Psi\left(  G\right)  $ and
$\left\vert S\right\vert \geq\left\vert N(S)\right\vert $, but $S\notin
Crown(G)$. However, $\{c,d\}\in\Psi\left(  G\right)  $ and $\left\vert
\{c,d\}\right\vert <\left\vert N(\{c,d\})\right\vert =\left\vert
\{u,v,w\})\right\vert $, but again $\{c,d\}\notin Crown(G)$.

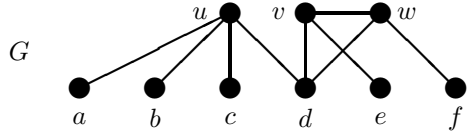
\begin{figure}[h]
\setlength{\unitlength}{1cm}\begin{picture}(5,1.5)\thicklines
\multiput(5,0.5)(1,0){6}{\circle*{0.29}}
\multiput(7,1.5)(1,0){3}{\circle*{0.29}}
\put(5,0.5){\line(2,1){2}}
\put(6,0.5){\line(1,1){1}}
\put(7,0.5){\line(0,1){1}}
\put(7,1.5){\line(1,-1){1}}
\put(8,0.5){\line(0,1){1}}
\put(8,1.5){\line(1,0){1}}
\put(8,1.5){\line(1,-1){1}}
\put(8,0.5){\line(1,1){1}}
\put(9,1.5){\line(1,-1){1}}
\put(5,0.1){\makebox(0,0){$a$}}
\put(6,0.1){\makebox(0,0){$b$}}
\put(7,0.1){\makebox(0,0){$c$}}
\put(8,0.1){\makebox(0,0){$d$}}
\put(9,0.1){\makebox(0,0){$e$}}
\put(10,0.1){\makebox(0,0){$f$}}
\put(6.6,1.5){\makebox(0,0){$u$}}
\put(7.65,1.5){\makebox(0,0){$v$}}
\put(9.35,1.5){\makebox(0,0){$w$}}
\put(4.2,1){\makebox(0,0){$G$}}
\end{picture}\caption{$G$ is a K\"{o}nig-Egerv\'{a}ry graph with
$\Omega\left(  G\right)  =$ $\left\{  \{a,b,c,d,e,f\}\right\}  $.}%
\label{fig1211}%
\end{figure}

By Theorem \ref{ThGreed}\textit{(iii)} and Theorem \ref{th1111}, we obtain the following.

\begin{corollary}
\label{cor9}If $G$ is a triangle-free graph with a unique perfect matching,
and $G[N[S]]$ is a K\"{o}nig-Egerv\'{a}ry graph with a perfect matching for
every $S\in\Psi\left(  G\right)  $, then $CritIndep(G)$ is a greedoid.
\end{corollary}

\section{Conclusions}

If $CritIndep(G)=Crown(G)$, then $d\left(  G\right)  =0$, because
$\emptyset\in Crown(G)$.

\begin{problem}
Characterize graphs satisfying $CritIndep(G)=Crown(G)$.
\end{problem}

Lemma \ref{Lemma1}, Theorem \ref{th3} and Corollary \ref{cor9} motivate the following.

\begin{conjecture}
Let $G$ be a triangle-free graph. $CritIndep(G)$ is a greedoid if and only if
$G[N[S]]$ is a K\"{o}nig-Egerv\'{a}ry graph with a unique perfect matching for
every $S\in\Psi\left(  G\right)  $.
\end{conjecture}

\begin{theorem}
\label{th11111}If $S_{1},S_{2}\in\Psi(G)$ and $N\left[  S_{1}\right]
\subseteq N\left[  S_{2}\right]  $, then there exist $S_{3}\subseteq
S_{2}-S_{1}$, such that $S_{1}\cup S_{3}\in\Psi(G)$ and $\left\vert S_{1}\cup
S_{3}\right\vert =\left\vert S_{2}\right\vert $.
\end{theorem}

\begin{proof}
Since $S_{1}\in\Psi(G)$, it follows that $\left\vert S_{2}\cap N\left[
S_{1}\right]  \right\vert \leq|S_{1}|$. Consequently, $S_{3}=S_{2}-\left(
S_{2}\cap N\left[  S_{1}\right]  \right)  $ is independent and $\left\vert
S_{3}\right\vert \geq\left\vert S_{2}\right\vert -\left\vert S_{1}\right\vert
$. Hence, we get that $S_{1}\cup S_{3}$ is independent. Moreover,
\[
\left\vert S_{2}\right\vert \geq\left\vert S_{1}\cup S_{3}\right\vert
=\left\vert S_{1}\right\vert +\left\vert S_{3}\right\vert \geq\left\vert
S_{1}\right\vert +\left\vert S_{2}\right\vert -\left\vert S_{1}\right\vert
=\left\vert S_{2}\right\vert ,
\]
because $N\left[  S_{1}\right]  \subseteq N\left[  S_{2}\right]  $.

Therefore, $\left\vert S_{2}\right\vert =|S_{1}\cup S_{3}|$. In addition,
$S_{1}\cup S_{3}\in\Psi(G)$ follows from the fact that $N\left[  S_{1}\cup
S_{3}\right]  \subseteq N\left[  S_{2}\right]  $.
\end{proof}

According to Theorem \ref{th111}, every $A\in\Psi(G)$ can be enlarged to some
maximum independent set. By Theorem \ref{th11111}, for each $S\in\Omega(G)$
this enlargement can be implemented using only elements of $S$.

\begin{corollary}
\cite{LevMan2,NemhTro} (Augmentation version of Nemhauser's and Trotter's
Theorem). If $S_{1}\in\Psi(G)$ and $S_{2}\in\Omega(G)$, then there exists
$S_{3}\subseteq S_{2}-S_{1}$ such that $S_{1}\cup S_{3}\in\Omega(G)$.
\end{corollary}

This motivates the following.

\begin{problem}
Characterize graphs whose families of local maximum independent sets are augmentoids.
\end{problem}

\end{document}